\documentclass[twoside,11pt]{article}

\usepackage{jmlr}
\usepackage[ruled]{algorithm2e}
\usepackage{amsmath}
\usepackage{comment}
\usepackage{mathtools}
\usepackage{xcolor}
\DeclareMathOperator*{\argmax}{argmax}
\DeclareMathOperator*{\argmin}{argmin}

\usepackage{enumitem}
\usepackage{float}
\newcommand{\subscript}[2]{$#1 #2$}

\ShortHeadings{Repeated Games with Asymmetric Information}{L.C.Dinh, L. Tran-Thanh, T.D. Nguyen, and A.B. Zemkoho}
\firstpageno{1}

\begin{document}

\title{Last Round Convergence and No-Instant Regret in Repeated Games with Asymmetric Information}

\author{\name Le Cong Dinh \email L.C.Dinh@soton.ac.uk \\
        \name Long Tran-Thanh \email L.Tran-Thanh@soton.ac.uk  \\
       \addr School of Electronics and Computer Science\\
       University of Southampton\\
       United Kingdom
       \AND
       \name Tri-Dung Nguyen \email T.D.Nguyen@soton.ac.uk  \\
       \name Alain B. Zemkoho \email A.B.Zemkoho@soton.ac.uk  \\
       \addr School of Mathematical Sciences \& CORMSIS\\
       University of Southampton\\
       United Kingdom}
\editor{}
\maketitle
\begin{abstract}
This paper considers repeated games in which one player has more information about the game than the other players. In particular, we investigate repeated two-player zero-sum games where only the column player knows the payoff matrix $A$ of the game. Suppose that while repeatedly playing this game, the row player chooses her strategy at each round by using a no-regret algorithm to minimize her (pseudo) regret. We develop a no-instant-regret algorithm for the column player to exhibit last round convergence to a minimax equilibrium. We show that our algorithm is efficient against a large set of popular no-regret algorithms of the row player, including the multiplicative weight update algorithm, the online mirror descent method/follow-the-regularized-leader, the linear multiplicative weight update algorithm, and the optimistic multiplicative weight update.
\end{abstract}


\section{Introduction}
 
Repeated two-player zero-sum games form one of the most studied classes of repeated games in game theory. 
In this setting, thanks to Blackwell's famous approachability theorem, if a player's strategies are generated by algorithms (i.e., policies) with a special property called ``no-regret", one can prove that, on average, that player does not perform worse than the best-fixed strategy in hindsight.
A direct implication of this result is that if both players choose to play such no-regret algorithms, their average payoffs will converge to the game's minimax value. Put differently, the players' strategies will converge to a minimax equilibrium on average (see e.g.,  \cite{Nicolo06} or  \cite{Arora2012} for more details). It can also be easily shown that this (on-average) convergence holds independently from the prior information that each player has about the payoff matrix A. That is, no matter how much prior information a player has about the game, she cannot exploit the other player if the latter uses a no-regret algorithm.

In this paper, we consider a shift of interest and investigate whether it is more beneficial for the column player if she has more information about the underlying game (e.g., knows more about the payoff matrix A) than her opponent (i.e., the row player), and thus, can exploit this advantage. In particular, our first main result shows that such asymmetry in prior knowledge allows the column player to have `\emph{no-instant regret}', a stronger no-regret concept in which the regret function is defined by comparing with a much larger space of possible recourse actions in the previous rounds as will be formally defined in Section~\ref{primary knowledge}.

In addition to having no-instant regret, the second property that we investigate is with regard to the `\emph{stability}', or \emph{last round convergence}, in the dynamic of the repeated games with asymmetric information. This is motivated by the fact that changing strategies through repeated games might be undesirable. For example, changing the (mixed) strategy of a company will increase the cost of operation to implement the new mixed strategy (e.g., as a result of having to hire new equipment and employees). Therefore, the company often aims not only to maximize the revenue (i.e., the average payoff) but also to reduce the cost of operation by having a stable strategy. For another example, consider a government-owned company, for whom, along with the average benefit, keeping the market stable is one of the key goals in order to increase social welfare. Finally, in system design, the designer (the column player) will want the participant (the row player) to play a certain strategy so that the system is well behaved.  

In the online learning literature, maximizing the average payoff and achieving the system's stability are often viewed as conflicting goals. That is, if all the player in a system follows a selfish behaviour (an FTRL no-regret algorithm) to maximize their payoff, then the dynamic of the system could become chaotic, and last round convergence never happens (see, e.g.,  \cite{mertikopoulos2018} for more details). The question is, whether there is a way to achieve both no-regret and stability in a system? 

In this paper, we show that it is possible to exploit the information asymmetry to achieve both stability and no-instant regret. The intuition behind this result can be explained as follows: If the row player believes that the goal for both players is to maximize their average payoffs, then she will typically choose to play a no-regret algorithm to achieve good average performance.  
Being aware of this, the column player can now choose an algorithm that exploits this information as well as the prior knowledge about matrix A to have no-instant regret and last round convergence (e.g., by carefully guiding the learning dynamics of the row player).
We should note here, however, that it is not trivial how this can be efficiently done. For example, if the column player keeps playing the same strategy (i.e., the minimax equilibrium), then while the system might achieve stability as the strategy of the row player will converge to the best response, this is not a no-regret algorithm and therefore, far away from being a no-instant-regret algorithm.

However, in many situations, one of the players, says the column player, does not only care about her average payoff but also the system' stability. For example, changing the strategy of a company will increase the cost of operation (i.e., to hire new equipment and employees). Therefore, the company wants to have a good revenue (the average payoff) and a stable strategy to reduce the cost of operation. Another example can come from social welfare aspect of a state capital company. Along with the average benefit, keeping the strategy of the opponents stable will lay a foundation for a stable economy, thus increasing social welfare. In system design, beside the average performance of the system, in order to build another layer to the current structure, the stability of the currently involved players is vital. However, in the literature, there is an opposition between average payoff and the system's stability. That is, if all the player in a system follows a selfish behaviour (an FTRL no-regret algorithm) to maximize their payoff, then there will be chaos in the dynamic of the system \cite{mertikopoulos2018}. So, is there a way to achieve both no-regret and stability in a system?

In this paper, we investigate whether this shift of interest is more beneficial for the column player if she has more information about the underlying game (e.g., knows more about the payoff matrix A) than her opponent (i.e., the row player), and thus, can exploit this advantage.
In particular, we argue that such asymmetry in prior knowledge allows the column player to have a stronger no-regret algorithm (a no-instant regret). Furthermore, the new algorithm will lead to last round convergence to minimax equilibrium for both players, thus achieve stability for the system.

The intuition behind this argument can be explained as follows: If the row player believes that the goal for both players is to maximize their average payoffs, then she will typically choose to play a no-regret algorithm to achieve good average performance.  
Being aware of this, the column player can now choose an algorithm that exploits this information as well as the prior knowledge about matrix A to have no-instant regret and last round convergence (e.g., by carefully guiding the learning dynamics of the row player).
However, it is not trivial how this can be efficiently done. For example, if the column player keeps playing the same strategy (i.e., the minimax equilibrium), then the system will still achieve stability as the strategy of the row player will converge the best response to the chosen strategy of the column player. However, this simple strategy will not be a no-regret algorithm and therefore, far away from a desired no-instant-regret algorithm. 

\subsection{Our contributions}

Motivated by the abovementioned challenge, we propose a new algorithm, the first of its kind, that achieves no-instant-regret; that is, the regret compared to the best action in each round, for the column player in the case the row player follows a no-regret algorithm. In the general case, we introduce a method for the column player to have no-regret property against random strategies of the row player while still maintaining no-instant-regret property against no-regret algorithm of the row player. 

Secondly, while on-average convergence has been extensively studied, it is still an open question whether last round convergence can be achieved, especially when the row player is also playing a no-regret algorithm (see Section~\ref{subsection related work} for more details). Against this background, we show that our algorithm, called the \emph{L}ast \emph{R}ound \emph{C}onvergence in \emph{A}symmetric games algorithm  (LRCA), provably achieves last round convergence to a minimax equilibrium of the corresponding game. 
We prove that in our asymmetric information setting if the column player follows LRCA and the row player follows an algorithm from a wide set of common no-regret algorithms,
then last round convergence to the minimax equilibrium of the game can be achieved.

Overall this paper has two main contributions. First, by changing the setting of games with symmetric information to their asymmetric information counterpart, we propose an algorithm that leads to last round convergence in many situations, which were proved not to hold (i.e., there is no last round convergence) in symmetric information settings; 
see Section \ref{section last round convergence} for more details.
Second, we show that by using the algorithm, the column player can achieve no-instant-regret property; see Section \ref{no-instant-regret section} for more details. This answer the question of how to achieve both maximizing the average payoff and stability in a repeated game. 
\subsection{Related work}
\label{subsection related work}

It is well-known that if both players use no-regret algorithms, their average strategies converge to a minimax equilibrium with the convergence rate of $\mathcal{O}(T^{-1/2})$; cf. \cite{Freund99}. \cite{Daskalakis2011} and \cite{Rakhlin2013} have further improved this result by developing no-regret algorithms with near-optimal convergence rate of $\mathcal{O}(\frac{\log(T)}{T})$. However, despite the extensive literature on no-regret algorithms, these algorithms typically provide on-average convergence only, but not last round convergence. For example, \cite{Bailey2018} proved that in games with an interior Nash equilibrium point, if the players use the multiplicative weight update (MWU) algorithm, then the last round strategy converges to the boundary. In addition, \cite{mertikopoulos2018} showed that by using regularized learning, the system's behaviour is Poincare recurrent; that is, there is a loop in the strategy dynamics of the players. This undesirable feature causes many issues in game theory and applications, including unwanted cyclic behaviour in training Generative Adversarial Networks (GANs). Thus, a learning dynamic leading to last round convergence is of importance in the development of the field (see, e.g., \cite{Daskalakis_gans} for more details).  
Note that in a recent paper, \cite{Daskalakis2018c} managed to prove that if both players use the optimistic multiplicative weight update algorithm (OMWU), then we have last round convergence to the minimax equilibrium if this equilibrium point is unique. 
This last round convergence result also requires another restrictive assumption, namely: 
the step size of the update mechanism has to be calculated from the payoff matrix A of the game. %
Therefore, if the row player does not know the matrix A of the game, then OMWU cannot guarantee last round convergence (as it requires both players to know matrix A). Besides, if the row player plays different no-regret algorithms such that MWU or FTRL, which are widely used in many applications, then OMWU cannot lead to the last round convergence either. This raises the question of whether there could be a robust algorithm, when playing against different no-regret algorithms, converging at the last round to minimax equilibrium. 

\subsection{Key assumptions}
\label{subsection key assumptions}

To proceed with the development of this paper, we make the following two assumptions:
\begin{enumerate}[label=\subscript{A}{{\arabic*}}]
    \item \label{Assumption A1}: The column player knows the matrix $A$ of the game.
    \item \label{Assumption A2}: The row player follows a no-regret algorithm.
\end{enumerate}
The rationale of these assumptions can be explained as follows:
In Assumption \ref{Assumption A1}, we consider the situation of unfair two-player zero-sum games in which the column player knows the matrix $A$ of the game. This arises in many cases in practice. For example, in the security games domain, an attacker can store the feedback from past observations and analyze the behaviour of the system. Thus, the attacker could know matrix A of the game. 

Another example comes from the perspective of a new company who enters an existing business market. In this market, every strategy and payoff of the players are revealed. Therefore, when a new company enters the market, they can anticipate what their payoff for a particular action of their strategies is. Thus, the new incomer knows the matrix A of the game. We argue that the asymmetric game assumption might appear in many other applications, and hence the setting deserves attention from the online learning community.

Assumption \ref{Assumption A2} comes from the vanilla property of no-regret algorithms: without prior information, a player will not do worse than the best-fixed strategy in hindsight by following a no-regret algorithm. In this study, we allow the row player to deviate from a no-regret algorithm in a certain way; that is, she can choose a non-fixed learning rate. We also consider the full information feedback (see, e.g., \cite{Bailey2018}, \cite{Daskalakis2011}, \cite{Freund99}).\footnote{Note that the main focus of this paper is on the investigation of the benefit of having asymmetric information. Thus, the analysis of other feedback cases, such as bandit or semi-bandit, is out of scope and remains part of future work. }

Note that our setting differs from that of \cite{Daskalakis2018c} in the following ways: while our asymmetric information assumption is more restrictive than the setting of \cite{Daskalakis2018c}, we require neither the knowledge of the update step size nor the uniqueness of the minimax equilibrium. In addition, our result does not require the row player to follow the OMWU. As such, we argue that our result can be applied to more real-world applications, due to its more reasonable and realistic assumptions (see Section~\ref{subsection key assumptions} for more detailed discussions).

\subsection{Structure of the paper}
The remainder of the paper is structured as follows:
before presenting the main results, we first provide some preliminaries in Section \ref{primary knowledge}.
Section \ref{section last round convergence} studies the LRCA algorithm and proves how it can provide last round convergence under different assumptions.
We then investigate how the column player uses LRCA to gain a no-instant-regret algorithm in Section \ref{no-instant-regret section}.  Finally,
Section \ref{section conclusion} concludes the work.

\section{Preliminaries}\label{primary knowledge}
Consider a repeated two-player zero-sum game. This game is described by a payoff matrix $A$, where $A$ is an 
$n \times m$ non-zero matrix with entries in $[0,1]$. The rows and columns of $A$ represent the {\em pure} strategies of the {\em row} and {\em column}  players, respectively.
We define the set of feasible strategies of the row player, at round $t$, by $\Delta_n:=\left\{x_t\in \mathbb{R}^n\left|~\sum_{i=1}^nx_t(i)=1,\;\; x_t(i) \geq 0 \;\; \forall i \in \{1,...,n\} \right.\right\}$. The set of feasible strategies of the column player, denoted by $\Delta_m$, is defined in a similar way. At round $t$, if the row (resp. column) player chooses a mixed strategy $x_t \in \Delta_n$ (resp. $y_t \in \Delta_m$), then the row player's payoff is $-x_t^\top Ay_t$, while the column player's payoff is $x_t^\top Ay_t$. Thus, the row (resp. column) player aims to minimise (resp. maximise) the quantity $x_t^\top Ay_t$ (resp.  $x_t^\top Ay_t$). John von Neumann's minimax theorem, founding stone in zero-sum games (\cite{Neumann1928}), states that
\begin{equation} \label{minimax theorem}
\max_{y\in \Delta_m}\min_{x \in \Delta_n} x^\top Ay \quad = \quad \min_{x \in \Delta_n}\max_{y\in \Delta_m} x^\top Ay \quad = \quad v
\end{equation}
for some $v\in \mathbb{R}$. We call a point $(x^*, y^*)$ satisfying the minimax theorem \eqref{minimax theorem} \emph{the minimax equilibrium of the game}. For such a point, we have the following inequalities  

\begin{equation}\label{f(x) property}
    \begin{aligned}
    \max_{y \in \Delta_m} x^\top Ay \;\; \geq \;\; v, \quad
\min_{x \in \Delta_n}x^\top Ay \;\; \leq\;\;v.  
    \end{aligned}
\end{equation}
Throughout this paper, we use the notation
\[
f(x) \;\; := \;\; \max_{y \in \Delta_m} x^\top Ay.
\]
Since $A$ is a non-zero matrix with entries in $[0,1]$, we have $f(x) >0$. Furthermore, it is easy to show that if $f(x)=v$, then $x$ is a minimax equilibrium strategy. Similarly, if $\min_{x \in \Delta_n} x^\top Ay = v$, then $y$ is also a minimax equilibrium strategy.

Next, we define the concept of a  {\em no-instant regret} that will play an important role in this paper. 
\begin{definition}
Let $x_1, x_2, \ldots$ be a sequence of mixed strategies played by the row player. An algorithm of the column player that generates a sequence of mixed strategies $y_1, y_2, \ldots$ is called a {\em no-instant-regret} algorithm if we have 
\begin{equation*}
    \lim_{T \rightarrow \infty} \frac{IR_T}{T} = 0, \;\; \text{where} \;\; IR_T := \frac{1}{T} \sum_{t=1}^T  \left(\max_{y \in \Delta_m} x_t^\top A y- x_t^\top Ay_t\right).
\end{equation*}
\end{definition}
Note that the no-instant-regret property is considered stronger than the usual no-regret property since the regret function $IR_T$ is summed up over the regrets in each and every past round while the usual regret function $R_T= \max_{y \in \Delta_m} \frac{1}{T} \sum_{t=1}^T  {x_t}^\top A(y-y_t)$ is defined over the same best response to the average of all rounds (i.e., mathematically, the max operator in $IR_T$ is under the summation operator, which means we might have different optimal $y$ for different $t$, while the max operator in $R_T$ is outside). In the literature, the no-instant regret is desirable but is impossible to achieve with current state of the art algorithms in the adversarial symmetric setting. Thanks to the new asymmetric information counterpart, we can propose a new algorithm with the no-instant-regret property.

To conclude this section, it is important to mention that in this paper, we will use the Kullback-Leibler divergence to understand the behaviour of the row player's strategies.
\begin{definition}[\cite{KL1951}]\label{RE definition}
The relative entropy or K-L divergence between two vectors 
$X_1$ and $X_2$ in $\Delta_n$ is defined as $RE(X_1||X_2) = \sum_{i=1}^n X_1(i)\log \left(\frac{X_1(i)}{X_2(i)}\right).$
\end{definition}
The Kullback-Leibler divergence is always non-negative. Gibbs's inequality (\cite{Gibbs1970}) indicates that $RE(X_1||X_2)= 0$ if and only if $X_1=X_2$ almost everywhere.

\section{Last round convergence to minimax equilibrium} \label{section last round convergence}

We first start with the analysis of last round convergence in asymmetric information cases.
In particular, we present the Last Round Convergence of Asymmetric games algorithm (LRCA) for the column player. We then show that our algorithm is robust to many no-regret algorithms played by the row player, namely, MWU, OMD/FTRL, LMWU and OMWU (i.e., it provides last round convergence when played against these algorithms).
Under Assumption \ref{Assumption A1}, the column player knows the matrix $A$ of the game and thus can calculate a minimax equilibrium strategy $y^*$ and the value $v$ of the game using linear programming (although $y^*$ may not be unique). 

%


For a sequence of strategies $x_1,x_2,...$ played by the row player, the LRCA algorithm (in Algorithm \ref{LRCA algorithm 1}) for the column player can be described as follows: 
At each odd round, the column player plays the minimax equilibrium strategy, $y^*$, so that in the next round, she can not only predict the distance between the current strategy of the row player and a minimax equilibrium, but also prevent the row player from deviating the current strategy. Then, at the following even round, the column player chooses a strategy such that the feedback to the row player, $Ay_t$, is a direction towards a minimax equilibrium strategy of the row player. Depending on the distance between the current strategy of the row player and a minimax equilibrium (which is measured by $f(x_{t-1})-v$), the column player chooses a suitable step size so that the strategy of the row player will approach a minimax equilibrium.

The Algorithm \ref{LRCA algorithm 1} (LRCA) will work for a large set of learning rate, including the constant learning rate case. The simpler algorithms, such that ``fictitious play" or ``best response to the last feedback" will fail to converge in the simple case of constant learning rate and do not have the no-instant-regret property in Section \ref{no-instant-regret section}.

We will prove in the following subsections that if the column player follows the LRCA algorithm and the row player uses one of the aforementioned no-regret algorithms; we will achieve last round convergence to the minimax equilibrium. 
\begin{algorithm}[t]
\SetAlgoLined
\textbf{Input:} Current iteration $t$, past feedback $x_{t-1}^\top A$ of the row player\\
\textbf{Output:} Strategy $y_t$ for the column player\\
\If{$t=2k-1, \;   k \in \mathbb{N}$}{$y_t=y^*$}
\If{$t=2k, \; k \in \mathbb{N}$}{ $e_t : = \argmax_{e \in \{e_1,e_2,\dots e_m\}}x_{t-1}^\top Ae$; \quad $f(x_{t-1}): = \max_{y \in \Delta_m}x_{t-1}^\top Ay$ \\
$\alpha_t: = \frac{f(x_{t-1})-v}{\max~\left(\frac{n}{4},2\right)}$\\
$y_t: =(1-\alpha_t)y^*+\alpha_t e_t$} 
\caption{\emph{L}ast \emph{R}ound \emph{C}onvergence in \emph{A}symmetric algorithm (LRCA)}\label{LRCA algorithm 1}
\end{algorithm}

\subsection{Last round convergence under  MWU}
One of the most well-studied no-regret algorithms in the game theory literature is the multiplicative weight update (MWU), 
which can defined as follows:
\begin{definition}[\cite{Freund99}] \label{Mwudef}
Let $y_1,y_2,...$ be a sequence of mixed strategies played by the column player. The row player is said to follow the MWU algorithm if strategy $x_{t+1}$ is updated as follows:
\begin{equation*}\label{MWUstep1}
\begin{array}{l}
 x_{t+1}(i)=x_{t}(i)\frac{ e^{-\mu_t e_i^\top A y_t}}{Z_t},\; i \in \{1, \dots n\},\\
\text{where } 
\left\{
\begin{array}{l}
Z_t=\sum_{i=1}^n x_t(i) e^{-\mu_t e_i^\top A y_t}, \mu_t \in [0,\infty) \; \text{is a parameter,} \\
e_i,\; i \in \{1, \ldots, n\}, \;\mbox{ is the unit-vector with } 1 \mbox{ at the } ith \mbox{ component.} 
\end{array}
\right.
\end{array}
\end{equation*}
\end{definition}
\cite{Bailey2018} proved that if both players follow the MWU then in the case of interior minimax equilibrium, the strategies will move away from the equilibirium and towards the boundary. In this subsection, we prove that Algorithm \ref{LRCA algorithm 1}(LRCA) played by the column player will lead to last round convergence in the case of MWU. The following lemma shows that the relative entropy between strategy of the row player and the minimax equilibrium is non-increasing.

\begin{lemma}\label{MWU K-L distance Lemma}
Assume that the row player follows the MWU algorithm with a non-increasing step size $\mu_t$ such that there exists  $t' \in \mathbb{N}$ with $\mu_{t'} \leq 1$. If the column player follows the Algorithm \ref{LRCA algorithm 1} (LRCA) then
\begin{equation*}
RE\left(x^*||x_{2k-1}\right)-RE\left(x^*||x_{2k+1}\right) \geq   \frac{1}{2}\mu_{2k}\alpha_{2k}(f(x_{2k-1})-v)\;\; \forall k \in \mathbb{N}: \;\; 2k\geq t',
\end{equation*}
where $RE$ denotes the relative entropy, which is defined in Definition 
\ref{RE definition}.
\end{lemma}

This Lemma (see Appendix \ref{proof K-L distance new} for its proof) can be used to prove the next result. 

\begin{theorem}\label{LRCA algorithm main proof}
Let $A$ be an $n \times m$ non-zero matrix with entries in $[0,1]$.
Assume that the row player follows the MWU algorithm with a non-increasing step size $\mu_t$ such that  $\lim_{T \to \infty}\sum_{t=1}^T  \mu_t=\infty$ and there exists $t' \in \mathbb{N}$ with $\mu_{t'} \leq 1$. If the column player plays Algorithm \ref{LRCA algorithm 1} (LRCA) then there exists a minimax equilibrium ${\bar{x}}^*$, such that $\lim_{t \to \infty}RE({\bar{x}}^*||x_t)=0$ and thus $lim_{t \to \infty}\; x_t= {\bar{x}}^*$ almost everywhere.
\end{theorem}
\begin{proof}
Let $x^*$ be a minimax equilibrium strategy of the row player ($x^*$ may not be unique). Since $\mu_t$ is a non-increasing step size, there exists $t'$ such that $\mu_t \leq 1$ for all $t\geq t'$.
Following Lemma \ref{MWU K-L distance Lemma}, for all $k \in \mathbb{N}$ such that $2k\geq t'$, we have
\begin{equation}\label{MWU3 1st important step}
RE(x^*||x_{2k+1})-RE(x^*||x_{2k-1}) \leq   -\frac{1}{2}\mu_{2k}\alpha_{2k}(f(x_{2k-1})-v).
\end{equation}
Thus, the sequence of relative entropy $RE(x^*||x_{2k-1})$ is non-increasing  for all $k \geq \frac{t'}{2}$. Since the sequence is bounded below by 0, it has a limit for any minimax equilibrium strategy $x^*$.

Since $t'$ is a finite number and $\sum_{t=1}^\infty \mu_t=\infty$, then $\sum_{t=t'}^\infty \mu_t=\infty$. Hence, 
    \[\lim_{T\to \infty}\sum_{k=\left \lceil{\frac{t'}{2}}\right \rceil}^{T}\mu_{2k} = \infty.\]
We will prove that $\forall \epsilon >0,\; \exists h \in \mathbb{N}$ such that following Algorithm \ref{LRCA algorithm 1} (LRCA) for the column player and MWU algorithm for the row player, the row player will play strategy $x_h$ at round h and $f(x_h)-v \leq \epsilon$. Suppose, for the purpose of deriving a contradiction, that $\exists \epsilon >0$ such that $\forall h \in \mathbb{N},\; f(x_h)-v > \epsilon$. Then $\forall k \in \mathbb{N}$, 
\[
\alpha_{2k}(f(x_{2k-1})-v)= \frac{(f(x_{2k-1})-v)^2}{\max\,(\frac{n}{4},2)} > \frac{e^2}{\max\,(\frac{n}{4},2)}.\]
Take k from $\left \lceil{\frac{t'}{2}}\right \rceil$  to T in the equation (\ref{MWU3 1st important step}) and sum them up side by side, we obtain: 
\begin{equation*}
    \begin{aligned}
    RE(x^*||x_{2T+1})&\leq RE(x^*||x_{t'}) - \frac{1}{2}\sum_{k=\left \lceil{\frac{t'}{2}}\right \rceil}^{T} \mu_{2k}\alpha_{2k}(f(x_{2k-1})-v) \\
    &\leq RE(x^*||x_{t'})-\frac{1}{2}\frac{e^2}{\max(\frac{n}{4},2)}\sum_{k=\left \lceil{\frac{t'}{2}}\right \rceil}^{T}\mu_{2k}.
    \end{aligned}
\end{equation*}
Since $\lim_{T\to \infty}\sum_{k=\left \lceil{\frac{t'}{2}}\right \rceil}^{T}\mu_{2k}= \infty$ and $RE(x^*||x_{T+1})\geq 0$, we arrive at a contradiction. 

Take a sequence of $\epsilon_k>0$ such that $\lim_{k \to \infty}\epsilon_k=0$. Then for each k, there exists $x_{t_k}\in \Delta_n$ such that $v \leq f(x_{t_k})\leq v+\epsilon_k.$ As $\Delta_n$ is a compact set and $x_{t_k}$ is bounded then following Bolzano-Weierstrass theorem, there is a convergence subsequence $x_{\bar{t}_k}$. The limit of that sequence, ${\bar{x}}^*$, is a minimax equilibrium strategy of the row player (since $f({\bar{x}}^*)=f(\lim_{k \to \infty} x_{\bar{t}_k})=\lim_{k \to \infty}f(x_{\bar{t}_k})=v$). Combining with the fact that $RE({{\bar{x}}^*}||x_{2k-1})$ is non-increasing for $k\geq \left \lceil{\frac{t'}{2}}\right \rceil$ and $RE({\bar{x}}^*||{\bar{x}}^*) =0$, we have $\lim_{k \to \infty}RE({\bar{x}}^*||x_{2k-1})=0$. We also note that 
\begin{equation*}
    \begin{aligned}
    RE({\bar{x}}^*||x_{2k})- RE({\bar{x}}^*||x_{2k-1})
    &=\mu_{2k-1} {{\bar{x}}^*}{} ^{\top} Ay_{2k-1}+\log
    \left(\sum_{i=1}^n {x_{2k-1}}(i)e^{-\mu_{2k-1}{e_i}^{\top} Ay^*}\right)\\
    &\leq \mu_{2k-1}v+\log \left(\sum_{i=1}^n {x_{2k-1}}(i)e^{-\mu_{2k-1}v}\right) =0,
    \end{aligned}
\end{equation*}
following the fact that ${x^*}^\top Ay \leq v$ for all $y \in \Delta_m$ and $x^\top Ay^* \geq v$ for all $x \in \Delta_n$. Thus, we have $\lim_{k \to \infty}RE({\bar{x}}^*||x_{2k})=0$ as well. Subsequently, 
$\lim_{t \to \infty}RE({\bar{x}}^*||x_{t}) =0$, which concludes the proof. 
\end{proof}

\begin{remark}
The optimal step size $\alpha_t$ in the case of MWU is $\alpha_t= \frac{f(x_{t-1})-v}{\mu_t f(x_{t-1})}$.
However, in order to make the Algorithm \ref{LRCA algorithm 1} (LRCA) robust against other algorithms of the row player, we choose the step size as shown in the algorithm. Follow 
Lemma \ref{MWU K-L distance Lemma} in the case of constant learning rate $\mu_t=\mu$, we have the complexity of the algorithm in order to achieve $f(x)-v \leq \epsilon$ is 
\[\frac{4\log(n)/\mu}{\epsilon^2}.\]
In Theorem \ref{LRCA algorithm main proof}, we impose the condition of $\mu_{t'} \leq 1$. The reason is that in the case of MWU with constant step size $\mu$, the regret will be $O(\mu)$ (\cite{Nicolo06}) so if the row player is strategic then she will never choose a big $\mu$. However, if the row player tries to outsmart the column player by playing a big step size in random occasions, then we can impose a factor $\log(t)$ in $\alpha_t$ as \[\alpha_t:= \frac{f(x_{t-1})-v}{\log(t) f(x_{t-1})}.\] If we have 
$\lim_{T \to \infty}\sum_{t=1}^T  \frac{\mu_t}{\log{t}}=\infty$,
then the algorithm with new step size will still converge to the minimax equilibrium in both constant or shrinking step size cases. However, this algorithm typically achieves slow convergence rate.
\end{remark}

\subsection{Last round convergence under  OMD/FTRL with Euclidean regularizer}
Another popular no-regret algorithm is online mirror descent (OMD). In the frame-work of a repeated two-player zero-sum game, the OMD with lazy projection and the follow-the-regularized-leader (FTRL) with Euclidean regularizer are the same (\cite{Shalev2012}): 
\begin{definition}
The row player is said to play the OMD/FTRL with Euclidean regularizer if the row player updates the strategy as follows:
\begin{equation*}
    x_t = \argmin_{x \in \Delta_n} x^\top  \left(\sum_{i=1}^{t-1}Ay_i\right)+\frac{1}{2\mu} x^\top x.
\end{equation*}
\end{definition}
\cite{mertikopoulos2018} showed that by using regularized learning, the system's behaviour is Poincare recurrent, and thus the last round strategy will not converge to a particular point. We will prove that, under mild condtions, the Algorithm \ref{LRCA algorithm 1} (LRCA) will lead to the last round convergence to the minimax equilibrium in the case of OMD/FTRL with Euclidean regularizer:
\begin{theorem}\label{OMD/FTRL}
Assume that the row player follows the OMD/FTRL with Euclidean regularizer and the column player follows the Algorithm \ref{LRCA algorithm 1} (LRCA). If there exists a ``fully-mixed" equilibrium strategy of the row player and the updated strategies of the row player are fully-mixed with the step size $\mu \leq 1$, then the Algorithm \ref{LRCA algorithm 1} (LRCA) leads to last round convergence to minimax equilibrium with the inequality 
\[||x_{2k-1}-x^*||^2-||x_{2k+1}-x^*||^2 \geq \alpha_t(f(x_{2k-1})-v)\; \; \forall k \in \mathbb{N}.
\]
\end{theorem}
The proof of Theorem \ref{OMD/FTRL} is provided in Appendix \ref{proof of OMD/FTRL}.
\subsection{Last round convergence under LMWU}
In this section, we study another type of multiplicative weight update algorithm:
\begin{definition}\label{LMWU definition}
The row player is said to play the LMWU if the row player updates the strategy as follows:
\begin{equation*}
    \begin{aligned}
        x_{t+1}(i)=\frac{x_t(i)(1-\mu_t {e_i}^\top Ay_t)}{\sum_{j=1}^n x_t(j)(1-\mu_t {e_j}^\top Ay_t) } \;\forall i \in \{1, \dots n\}.
    \end{aligned}
\end{equation*}
\end{definition}
Our next theorem will prove that the Algorithm \ref{LRCA algorithm 1} (LRCA) will work in the case of LMWU:
\begin{theorem}\label{proof for LMWU}
Assume that the row player follows the LMWU algorithm with a non-increasing step size $\mu_t$ such that: $\sum_{i=1}^n \mu_t = \infty \; , \; \lim_{t \to \infty} \mu_t =0.$ If the column player follows the Algorithm \ref{LRCA algorithm 1} (LRCA), then we have
\[ RE(x^*\|x_{2k-1})-RE(x^*\|x_{2k+1})\geq \frac{\mu_{2k} \alpha_{2k} (f(x_{2k-1})-v)}{2}\; \; \forall k \in \mathbb{N}. \]
Thus, there will be last round convergence to the minimax equilibrium.
\end{theorem}
\subsection{Last round convergence under OMWU}
Finally, we consider the optimistic multiplicative weight update(\cite{Daskalakis2018c}):
\begin{definition}
The row player is said to play the OMWU if the row player updates the strategy as follows:
\begin{equation*}
\begin{aligned}
    x_{t+1}(i)=x_t(i) \frac{e^{-2\mu e_i^\top Ay_t+ \mu e_i^\top Ay_{t-1}}}{\sum_{j=1}^n x_t(j) e^{-2\mu e_j^\top Ay_t+ \mu e_j^\top Ay_{t-1}} } \quad \forall i \in \{1 \dots  n\}.
\end{aligned}
\end{equation*}
\end{definition}
\cite{Daskalakis2018c} proves that if both players use OMWU then there will be last round convergence to minimax equilibrium. We prove that our modified version of LRCA has the same property. In the Algorithm \ref{LRCA algorithm 1} (LRCA), the column player just uses one ``stabilizing" strategies $y_{t-1}= y^*$ before exploiting the strategy of the row player. However, in the case of optimistic multiplicative weight update, as the row player uses the information about the last two round of the game, we need two stabilizing strategies. It will not change the result of the LRCA algorithm in other cases, but it will make the algorithm run slower to converge to the minimax equilibrium strategy. The reason is that the second stabilizing step does not make any effect on the algorithm where the row player only uses the latest feedback. In the case of OMWU, the algorithm is described as Algorithm \ref{LRCA algorithm with 2 stablizing} in Appendix \ref{proof of MWU detail}.
We then have the following result:
\begin{theorem}\label{OMWU proof}
Assume that the row player follows the OMWU algorithm. If the column player follows the Algorithm \ref{LRCA algorithm with 2 stablizing} in Appendix \ref{proof of MWU detail}, then there will be last round convergence to minimax equilibrium.
\end{theorem}
We include the full detail of the algorithm in the Appendix \ref{proof of MWU detail}
\begin{remark}
Algorithm \ref{LRCA algorithm 1} (LRCA) can work with other no-regret algorithms besides those four common ones considered in this paper should the no-regret algorithm of the row player has a ``stability" property as defined in Definition \ref{stability condition} in Appendix \ref{further last round convergence extension}. We note that this stability property holds for all the four aforementioned no-regret algorithms. We provide a proof for the case of the FTRL in Appendix \ref{further last round convergence extension} but this can be extended to the other three no-regret algorithms.
\end{remark}


\section{No-instant-regret algorithm}\label{no-instant-regret section}
In this section, we first show that if the column player wants to achieve both the no-regret and stability properties, then the row player's strategy needs to converge to the minimax equilibrium. Then, we show that our Algorithm \ref{LRCA algorithm 1} (LRCA) is a no-instant-regret algorithm for the column player when the row player follows the aforementioned no-regret algorithms. In a general case, we suggest a method to combine our LRCA algorithm with another no-regret algorithm (such that Adahedge \cite{Rooij2014}) so that the new algorithm will still have no-regret property against random sequences of the row player 
while maintaining no-instant-regret in the specific situation.  
\begin{theorem} \label{sufficient condition for stability and no-regret}
Suppose that the row player follows a common no-regret algorithm such as MWU, OMD, FTRL or LMWU. Then, the column player cannot achieve stability and the no-regret property if the row player's strategy does not converge to a minimax equilibrium of the game.
\end{theorem}
\begin{proof}
Suppose that the column player achieves both stability and no-regret property. The strategy of the column player will then converge, say to $\hat{y}$. Following the property of common no-regret algorithms, the strategy of the row player will also converge to a single best response $\hat{x}$ to $\hat{y}$:
\[\hat{x}=\argmin_{x \in \Delta_n}x^\top A \hat{y}.\]
Furthermore, since the strategy of the column player is no-regret, we must also have
\[\hat{y}=\argmax_{y \in \Delta_m}\hat{x}^\top A y.\]
Therefore, by definition, $(\hat{x}, \hat{y})$ is a minimax equilibrium of the game.
\end{proof}
Our Algorithm \ref{LRCA algorithm 1} (LRCA) satisfies the sufficient condition in Theorem \ref{sufficient condition for stability and no-regret}. Next, we prove the no-instant-regret property of the algorithm, clarifying the design of the LRCA-\ref{LRCA algorithm 1}.
\begin{theorem}\label{no-instant-regret theorem}
Assume that the row player follows any of these no-regret algorithms with any learning rate: MWU, OMD/FTRL with Euclidean regularizer and LMWU. If there exists a fully mixed minimax strategy for the row player, then by following 
Algorithm \ref{LRCA algorithm 1} (LRCA), the column player will achieve the no-instant-regret property with the instant-regret satisfying
\[R_T \leq IR_T = \mathcal{O}\left(\sqrt{n \log(n)} {T}^{3/4}\right).\]
Furthermore, in the case the row player uses a constant learning rate, we have
\[IR_T = \mathcal{O}\left(\sqrt{n \log(n)} {T}^{1/2}\right).\]
\end{theorem}
\begin{proof}
We prove the theorem in the case the row player follows the MWU algorithm. The proofs of other cases are given in Appendix \ref{no-instant-regret proof}.

For the odd round $2k-1$, the instant-regret of the column player at round $2k-1$ will satisfy
\[IR^{2k-1}=\max_{i \in {1,..m}} x_{2k-1}^\top A e_i- x_{2k-1}^\top A y^* \leq f_{2k-1}-v. \]
For the even round $2k$, considering  the existence of the fully mixed minimax equilibrium of the row player, we then have $Ay^*=vI_1$ ($I_1$ is a vector of all $1$ element) and thus  $x_{2k}=x_{2k-1}$. Therefore $IR^{2k} \leq f_{2k-1}-v.$

Combining the case of odd and even round, we derive
\[IR_T \leq 2 \sum_{k=1}^{T/2} (f_{2k-1}-v).\]
Now, following Lemma \ref{MWU K-L distance Lemma} in the case $n \geq 8$, we have
\begin{equation*}
    \begin{aligned}
    \frac{1}{2} \mu_{2k} \frac{(f(x_{2k-1})-v)^2}{n/4} \leq RE(x^*||x_{2k-1})-RE(x^*||x_{2k+1})\\
    \implies \sum_{k=1}^{T/2} \mu_{2k} (f(x_{2k-1})-v)^2 \leq \frac{n}{2} RE(x^*||x_1) \leq \frac{n\log(n)}{2}.
    \end{aligned}
\end{equation*}
Using the Cauchy–Schwarz inequality, we can then derive that
\begin{equation*}
    \begin{aligned}
        \sum_{k=1}^{T/2}(f(x_{2k-1})-v) \leq \sqrt{\frac{n\log(n)}{2}} \sqrt{\sum_{k=1}^{T/2} \frac{1}{\mu_{2k}}} \label{no-instant-regret proof b}
       \implies IR_T \leq \sqrt{2 n\log(n)} \sqrt{\sum_{k=1}^{T/2} \frac{1}{\mu_{2k}}}.
    \end{aligned}
\end{equation*}
If the row player follows the constant step size $\mu$, then we have
\[IR_T \leq \frac{\sqrt{n \log(n)}}{\sqrt{\mu}} T^{1/2}.\]
If the row player follows a decreasing step size $\mu_k = \sqrt{8 \log(n)/k}$ (\cite{Nicolo06}) to make the algorithm no-regret, then we have
\[IR_T \leq \frac{1}{2} n^{1/2} \log(n)^{1/4} T^{3/4}.\]
Indeed, for any sequence of step size $\mu_k$ such that $\sum_{k=1}^{T/2} \frac{1}{\mu_{2k}} \leq T^{3/2} $, the theorem holds.
\end{proof}
In the case of row player uses constant learning rate, Algorithm \ref{LRCA algorithm 1} achieves the average instant regret of $\mathcal{O}(T^{-1/2})$, better than state of the art no-regret algorithms which obtain the same average but in the normal regret $R_T$.
\begin{remark}
In order to make the LRCA-\ref{LRCA algorithm 1} robust, the step size $\alpha$ is not chosen optimal in the specific case of MWU. If we choose the optimal step size $\alpha_k= \frac{f(x_{k-1}-v)}{f(x_{k-1})}$ in the case of MWU, then we achieve a tighter bound of ${IR}_T =O(log(n) T^{3/4}).$
\end{remark}
In the general case where the column player does not know whether the row player will follow a no-regret algorithm, she can follow the following Algorithm \ref{LRCA in general case} to achieve no-regret algorithm in any situations while maintaining the no-instant-regret property against no-regret algorithm of the row player. The idea is to put the LCRA-\ref{LRCA algorithm 1} on top of another no-regret algorithm. When the regret of LCRA-\ref{LRCA algorithm 1} at a certain time exceeds a threshold, then we swap to the chosen algorithm. If the row player follows a no-regret algorithm then the LRCA-\ref{LRCA algorithm 1} regret will never exceed the threshold; thus we will have no-instant regret. By doing that, the column player sacrifices the optimal rate of no-regret in the worst case in order to achieve a much better no-instant regret in the case the row player follows a no-regret algorithm.  

\begin{algorithm}[t]
\SetAlgoLined
\textbf{Input:} Current iteration $t$, past feedback $x_{t-1}^\top A$ of the row player, total regret up to time $t: R_t$\\ 
\textbf{Output:} Strategy $y_t$ for the column player\\
\eIf{$R_t \leq \sqrt{n\log(n)}{t}^{3/4} $}{Follow the Algorithm \ref{LRCA algorithm 1} (LRCA)}{Follow Adahedge algorithm \cite{Rooij2014} onwards}
\caption{Combination of LRCA and Adahedge algorithm}\label{LRCA in general case}
\end{algorithm}
Algorithm \ref{LRCA in general case} will have the regret $R_T=O(\sqrt{n}{T}^{3/4})$ against random sequence strategies of the row player while maintain no-instant regret against the no-regret algorithm of the row player.

To sum up, we have shown that our LRCA-\ref{LRCA algorithm 1} would have a strong no-instant-regret property while playing against a large set of no-regret algorithms. More generally, we have proposed Algorithm \ref{LRCA in general case} with no-regret property while maintaining no-instant regret against the row player's no-regret algorithm.

\section{Conclusion}\label{section conclusion}
In this paper, we have proved that our algorithm \ref{LRCA algorithm 1} (LRCA)  leads to last round convergence to minimax equilibrium in many no-regret algorithms played by the row player, including MWU, OMD/FTRL, LMWU and OMWU. This answered the open question raised in \cite{Bailey2018} whether there is a learning dynamics leading to last round convergence rather than average. We have also shown that the column player can achieve a no-instant-regret algorithm in the asymmetric setting.
A future research direction is to improve our LRCA algorithm so that it can work well in other no-regret algorithms, while assuming that the column player only knows part of the matrix $A$. Another direction is to consider dynamics in bandit or semi-bandit feedback settings.

\bibliography{main}

\clearpage

\appendix
\section{Further last round convergence results}\label{further last round convergence extension}

Algorithm \ref{LRCA algorithm 1} (LRCA) can work with other no-regret algorithms besides those four common ones suggested in Section \ref{section last round convergence} should the no-regret algorithm of the row player has  a ``stability" property as formally defined below
\begin{definition} \label{stability condition}
We call a no-regret algorithm played by the row player have ``stability" property if:\[y_t=y^* \;\implies\;x_{t+1}=x_t \; \forall t.\]
\end{definition}
We then have the following theorem:
\begin{theorem}\label{extension of LRCA algorithm}
Assume that the row player follows a no-regret algorithm with the ``stability property". Then, by following the Algorithm \ref{LRCA algorithm 1} (LRCA),  for all $\epsilon >0$, there exists $l \in \mathbb{N}$ such that 
\[  f(x_l)- v \leq \epsilon.
\]
\end{theorem}
\begin{proof}
We will prove the theorem by contradiction. Suppose there exists $\epsilon > 0$ such that: 
\[ f(x_l) - v > \epsilon, \; \forall l \in \mathbb{N}.
\]
Then, follow the update rule of Algorithm \ref{LRCA algorithm 1} (LRCA) we have:
\[y_{2k-1}=y^* \; ; \; \alpha_{2k} = \frac{ f(x_{2k-1}) - v}{\max(\frac{n}{4}, 2)} > \frac{\epsilon}{\max(\frac{n}{4}, 2)}.
\]
By the stability property, as $y_{2k-1} = y^*$, we then have: $x_{2k-1}= x_{2k}$. Follow the update rule of Algorithm \ref{LRCA algorithm 1} (LRCA) :
\begin{subequations}
    \begin{align}
     x_{2k}^\top Ay_{2k} &= x_{2k-1}^\top  A \left( (1-\alpha_{2k})y^*+ \alpha_{2k}e_{2k} \right)\nonumber \\
     &\geq (1-\alpha_{2k}) v +\alpha_{2k} f(x_{2k-1}) \label{Stablity 1 a} \\
     &> (1-\alpha_{2k}) v +\alpha_{2k} (v+\epsilon) \label{Stablity 1 b} \\
     &\geq v + \frac{\epsilon^2}{\max (\frac{n}{4},2 )} \nonumber,
    \end{align}
\end{subequations}
Where Inequality (\ref{Stablity 1 a}) is due to 
\[x^\top Ay^* \geq v \; \forall x \in \Delta_n,\]
and where Inequality (\ref{Stablity 1 b}) comes from the assumption of $\epsilon$.
We then have: 
\[\frac{1}{T} \sum_{t=1}^T  x_t^\top Ay_t \geq \frac{v+ \left( v + \frac{\epsilon^2}{\max (\frac{n}{4},2 )}\right)}{2} = v+ \frac{\epsilon^2}{2\max (\frac{n}{4},2 )} .\]
We also note that from the definition of the value of the game: 
\[\min_{i}\frac{1}{T} \sum_{t=1}^T  e_i^\top Ay_t = \min_{i}e_i^\top A  \frac{\sum_{t=1}^T y_t}{T} \leq v.
\]
Thus, we have: 
\[
\lim_{T \rightarrow \infty} {\min_{i}\frac{1}{T} \sum_{t=1}^T  e_i^\top Ay_t-\frac{1}{T} \sum_{t=1}^T  x_t^\top Ay_t} \leq v - \left( v+ \frac{\epsilon^2}{2\max (\frac{n}{4},2 )} \right) = - \frac{\epsilon^2}{2\max (\frac{n}{4},2 )},  \]
contradicting to the definition of a no-regret algorithm: 
\[
\lim_{T \rightarrow \infty} {\min_{i}\frac{1}{T} \sum_{t=1}^T  e_i^\top Ay_t-\frac{1}{T} \sum_{t=1}^T  x_t^\top Ay_t}=0. \]
\end{proof}
There are many no-regret algorithms with stability property. In the next theorem, we will prove a class of FTRL has the stability property.
\begin{theorem} \label{FTRL has stability property}
Assuming that the row player follows the FTRL algorithm with a regularizer  $R(x)$:
\[x_t = \argmin_{x \in \Delta_n} x^\top  \left(\sum_{i=1}^{t-1} Ay_i\right) + R(x).\]
If there exists a fully-mixed minimax equilibrium strategy for the row player, then the FTRL algorithm has stability property.
\end{theorem}
\begin{proof}
As there exists a fully-mixed minimax equilibrium strategy of the row player, we have $Ay^*=vI_1$, where $I_1$ is a vector of all 1 element. Thus, we have:
\[x^TAy^*=v \;\forall x \in \Delta_n.\]
When the column player follows the minimax strategy, the minimization for $x_t$ and $x_{t+1}$ only differ in a constant term $v$, so their solutions are the same.
\end{proof}
\section{Proofs}

\subsection{Proof of Lemma \ref{MWU K-L distance
Lemma}}\label{proof K-L distance new}
Following the Definition \ref{RE definition} of relative entropy we have:
\begin{equation*}
\begin{aligned}
    &RE(x^*||x_{2k+1})-RE(x^*||x_{2k-1})\nonumber \\
   &= \left(RE(x^*||x_{2k+1})-RE(x^*||x_{2k})\right)+\left(RE(x^*||x_{2k})-RE(x^*||2k-1)\right) \nonumber\\
   &=\left(\sum_{i=1}^n x^*(i)\log\left(\frac{x^*(i)}{x_{2k+1}(i)}\right)- \sum_{i=1}^n x^*(i)\log\left(\frac{x^*(i)}{x_{2k}(i)}\right)\right) + \\
   &\quad \left( \sum_{i=1}^n x^*(i)\log\left(\frac{x^*(i)}{x_{2k}(i)}\right)- \sum_{i=1}^n x^*(i)\log\left(\frac{x^*(i)}{x_{2k-1}(i)}\right)\right)\\
   &= \left(\sum_{i=1}^n x^*(i)\log\left(\frac{x_{2k}(i)}{x_{2k+1}(i)}\right) \right)+ \left(\sum_{i=1}^n x^*(i)\log\left(\frac{x_{2k-1}(i)}{x_{2k}(i)}\right)\right).
\end{aligned}    
\end{equation*}
Following the update rule of the multiplicative weight update algorithm in Definition \ref{Mwudef} we have:
\begin{subequations}\label{mwu1}
\begin{align}
        &RE(x^*||x_{2k+1})-RE(x^*||x_{2k-1}) \nonumber\\
        &=\left(\mu_{2k} {x^*}^\top  Ay_{2k}+\log(Z_{2k})\right) + \left(\mu_{2k-1} {x^*}^\top  Ay_{2k-1}+\log(Z_{2k-1})\right) \nonumber\\
        &\leq \left(\mu_{2k} v + \log\left(\sum_{i=1}^n x_{2k}(i)e^{-\mu_{2k} e_i^\top Ay_{2k}}\right)\right)+ \left(\mu_{2k-1} v +\log(Z_{2k-1})\right) \label{MWU1a}\\
        &=\left(\mu_{2k} v + \log\left(\sum_{i=1}^n x_{2k-1}(i)e^{-\mu_{2k-1} e_i^\top Ay_{2k-1}}e^{-\mu_{2k} e_i^\top Ay_{2k}}\right)-\log(Z_{2k-1})\right) \nonumber\\
        &+ \left(\mu_{2k-1} v +\log(Z_{2k-1})\right),\nonumber
\end{align}
\end{subequations}
where Inequality (\ref{MWU1a}) is due to the fact that ${x^*}^\top Ay \leq v \; \forall y \in \Delta_m$. Thus,
\begin{subequations}
    \begin{align}
    &RE(x^*||x_{2k+1})-RE(x^*||x_{2k-1}) \nonumber\\
       &\leq \left(\mu_{2k} v + \log\left(\sum_{i=1}^n x_{2k-1}(i)e^{-\mu_{2k-1} e_i^\top Ay^*} e^{-\mu_{2k} e_i^\top Ay_{2k}}\right)\right)+ \mu_{2k-1} v \nonumber \\
        &\leq \left(\mu_{2k} v + \log\left(\sum_{i=1}^n x_{2k-1}(i)e^{-\mu_{2k-1} v} e^{-\mu_{2k} e_i^\top Ay_{2k}}\right)\right)+ \mu_{2k-1} v \label{MWU1b}\\
        &= \mu_{2k} v + \log\left(\sum_{i=1}^n x_{2k-1}(i) e^{-\mu_{2k} e_i^\top Ay_{2k}}\right) \nonumber,
    \end{align}
\end{subequations}
where Inequality (\ref{MWU1b}) is the result of the inequality:
\[x^\top Ay^* \geq v \quad \forall x \in \Delta_n.\]
Now, using the update rule of Algorithm \ref{LRCA algorithm 1} (LRCA)
\[y_{2k}=(1-\alpha_{2k})y^*+\alpha_{2k}e_{2k},\]
we then have:
\begin{subequations}
    \begin{align}
&RE(x^*||x_{2k+1})-RE(x^*||x_{2k-1}) \nonumber\\
&\leq \mu_{2k} v + \log\left(\sum_{i=1}^n x_{2k-1}(i) e^{-\mu_{2k} e_i^\top Ay_{2k}}\right)\nonumber\\
&= \mu_{2k} v + \log\left(\sum_{i=1}^n x_{2k-1}(i) e^{-\mu_{2k} e_i^\top A\left((1-\alpha_{2k})y^*+\alpha_{2k}e_{2k}\right)}\right)\nonumber\\
&\leq \mu_{2k} v + \log\left(\sum_{i=1}^n x_{2k-1}(i) e^{-\mu_{2k} ((1-\alpha_{2k})v+e_i^\top A(\alpha_{2k}e_{2k}))}\right).\label{MWU3 proof 1a}
    \end{align}
\end{subequations}
The Inequality (\ref{MWU3 proof 1a}) holds as:  \[\quad x^\top Ay^* \geq v \quad \forall x \in \Delta_n.\]
This leads to 
\begin{subequations}\label{MWU3 proof 1}
\begin{align}
&RE(x^*||x_{2k+1})-RE(x^*||x_{2k-1}) \nonumber \\
& \leq \mu_{2k} \alpha_{2k} v +\log\left(\sum_{i=1}^n x_{2k-1}(i) e^{-\mu_{2k}\alpha_{2k} e_i^\top Ae_{2k}}\right) \nonumber\\
&\leq \mu_{2k} \alpha_{2k} v + \log\left(\sum_{i=1}^n x_{2k-1}(i)(1-(1-e^{-\mu_{2k}\alpha_{2k}}){e_i}^\top Ae_{2k}\right)\label{MWU3 proof 1b} \\
&=\mu_{2k} \alpha_{2k} v + \log\left(1-(1-e^{-\mu_{2k}\alpha_{2k}}){x_{2k-1}}^\top Ae_{2k} \right)\nonumber\\
&\leq \mu_{2k} \alpha_{2k} v - (1-e^{-\mu_{2k}\alpha_{2k}}){x_{2k-1}}^\top Ae_{2k} \label{MWU3 proof 1c}\\
&=\mu_{2k} \alpha_{2k} v -(1-e^{-\mu_{2k}\alpha_{2k}}) f(x_{2k-1})\nonumber,
\end{align}
\end{subequations}
where Inequalities $(\ref{MWU3 proof 1b}, \ref{MWU3 proof 1c})$ are due to
\[\beta^x \leq 1-(1-\beta)x \quad \forall \beta \geq 0 \; x \in [0,1] \; \text{and} \; \log(1-x) \leq -x \; \; \forall x < 1.\]
We can develop Inequality (\ref{MWU3 proof 1c}) further as 
\begin{subequations}
\begin{align}
&RE(x^*||x_{2k+1})-RE(x^*||x_{2k-1}) \nonumber\\
&\leq \mu_{2k} \alpha_{2k} v -\left(1-e^{-\mu_{2k}\alpha_{2k}}\right)f(x_{2k-1}) \nonumber\\
&\leq \mu_{2k} \alpha_{2k} v -\left(1-\left(1-\mu_{2k}\alpha_{2k} +\frac{1}{2}(\mu_{2k}\alpha_{2k})^2\right)\right)f(x_{2k-1})\label{MWU3 proof 2a}\\
&=-\mu_{2k}\alpha_{2k}(f(x_{2k-1})-v) +\frac{1}{2}(\mu_{2k} \alpha_{2k})^2 f(x_{2k-1})\nonumber\\
&\leq -\mu_{2k}\alpha_{2k}(f(x_{2k-1})-v)+\frac{1}{2}\mu_{2k}\alpha_{2k}\mu_{2k} \frac{f(x_{2k-1})-v}{f(x_{2k-1})}f(x_{2k-1}) \label{MWU3 proof 2c}\\
&\leq  -\mu_{2k}\alpha_{2k}(f(x_{2k-1})-v)+\frac{1}{2}\mu_{2k}\alpha_{2k}\ (f(x_{2k-1})-v) \label{MWU3 proof 2b}\\
&=-\frac{1}{2}\mu_{2k}\alpha_{2k}(f(x_{2k-1})-v) \leq 0 \nonumber .   
\end{align}
\end{subequations}
Here, Inequality ($\ref{MWU3 proof 2a} $) is due to $e^x \leq 1+x+\frac{1}{2}x^2\;\; \forall x \in [-\infty,0]$, Inequality (\ref{MWU3 proof 2c}) comes from the definition of $\alpha_{t}$:
\[\alpha_t= \min \left( \frac{f(x_{t-1})-v}{f(x_{t-1})}, \frac{f(x_{t-1})-v}{n/4}, \frac{f(x_{t-1})-v}{2} \right ) = \frac{f(x_{t-1})-v}{\max (\frac{n}{4},2)},
\]
and finally Inequality ($\ref{MWU3 proof 2b} $) comes from the choice of k at the beginning of the proof, i.e., $\mu_{2k} \leq 1$.

\subsection{Proof of Theorem \ref{OMD/FTRL}} \label{proof of OMD/FTRL}

Let us denote $-\sum_{i=1}^{t-1}Ay_i = \theta_t$. Then we can rewrite the expression of $x_t$ as
\begin{equation*}
\begin{aligned}
    x_t= \argmin_{x \in \Delta_n} -x^\top \theta_t +\frac{1}{2\mu} x^\top x = \argmin_{x \in \Delta_n} \| x-\mu \theta_t \|^2.
\end{aligned}
\end{equation*}
When the updated strategies of the row player are fully-mixed, we can write the Lagrangian function as:
\[\mathbb{L}(x,\hat{\mu}, \lambda) =\| x-\mu \theta_t \|^2 + \sum_{i=1}^n \hat{\mu}_i(-x(i)) + \lambda ( \sum_{i=1}^n x(i) - 1 ). \]
The KKT conditions (\cite{KKT1951}) of the above problem becomes:
\begin{equation*}
\begin{aligned}
 &2(x-\mu \theta_t)-\hat{\mu} + \lambda e = 0, \\
 &\hat{\mu}_i x(i) = 0 \; \forall i \in \{1, \dots, n\},\\
 &\sum_{i=1}^n x(i) =1, \\
 &\text{where $e =(1,\ldots, 1)^\top$ is the unit vector of size n.}
\end{aligned}
\end{equation*}
Since our assumption about fully-mixed update strategy, $x_t(i) \neq 0 \; \forall i \in \{1, \dots, n\}$. Thus,
\[\hat{\mu}_i=0 \quad \forall i \in \{1, \dots, n\}.\]
Therefore, the update strategy of the row player at round $t$ will be:
\begin{equation}\label{Update of online mirror descent}
x_t(i) = \frac{n \mu \theta_t(i)-\mu \sum_{j=1}^n\theta_t(j)+1}{n} \quad \forall i \in \{1,2,\dots n\}.
\end{equation}
From Equation (\ref{Update of online mirror descent}), for any strategy $x \in \Delta_n$, we have:
\begin{equation*}
    \begin{aligned}
        &(x_t-\mu \theta_t)^\top  (x_t-x) = 0 \\
        &\implies \|x^*-x_t\|^2=\|x^*-\mu\theta_t\|^2- \|x_t-\mu\theta_t\|^2.
    \end{aligned}
\end{equation*}
Similarly, we can prove that if the updated strategy $x_{t+1}$ is fully-mixed then:
\begin{equation*}
    \begin{aligned}
     & (x_{t+1}-\mu\theta_{t+1})^ \top (x_{t+1}-x) =0\; \forall x \in \Delta_n \\
     &\implies \|x_{t-1}-x_{t+1}\|^2 = \|\mu \theta_{t+1}-x_{t-1}\|^2- \|\mu \theta_{t+1}-x_{t+1}\|^2.
    \end{aligned}
\end{equation*}
We then have for any minimax equilibrium strategy $x^*$:
\begin{equation*}
    \begin{aligned}
        & \|x_{t-1}-x^*\|^2 -\|x_{t+1}-x^*\|^2 \nonumber\\
        & = \left(\|\mu \theta_{t-1}-x^*\|^2- \|\mu \theta_{t-1}-x_{t-1}\|^2\right)-\left(\|\mu \theta_{t+1}-x^*\|^2- \|\mu \theta_{t+1}-x_{t+1}\|^2\right) \nonumber\\
       & =\left(\|\mu\theta_{t+1}-x^*+\mu (Ay_{t-1}+Ay_t)\|^2- \|\mu \theta_{t+1}-x_{t-1} +\mu (Ay_{t-1}+Ay_t)\|^2\right)\\
       &-\left(\|\mu \theta_{t+1}-x^*\|^2- \|\mu \theta_{t+1}-x_{t+1}\|^2\right)\\
       & =2\mu (x_{t-1}-x^*)^\top (Ay_{t-1}+Ay_t)-\left(\|\mu \theta_{t+1}-x_{t-1}\|^2- \|\mu \theta_{t+1}-x_{t+1}\|^2\right)\\
       & = 2\mu (x_{t-1}-x^*)^\top (Ay_{t-1}+Ay_t) - \|x_{t-1}-x_{t+1}\|^2.
    \end{aligned}
\end{equation*}
By the Algorithm \ref{LRCA algorithm 1} (LRCA), when $t$ is even we have: $y_{t-1}= y^*, y_t= (1- \alpha_t) y^* + \alpha_t e_t$. Furthermore, as we have a fully-mixed equilibrium for the row player, it follows that $x^\top Ay^*= v \quad \forall x \in \Delta_n$. Therefore, we have: 
\begin{subequations}
    \begin{align}
   & \|x_{t-1}-x_{t+1}\|^2 = \sum_{i=1}^n (x_{t-1}(i)-x_{t+1}(i))^2\nonumber\\
&= \sum_{i=1}^n \left(\mu \left((Ay_{t-1}+Ay_t)(i)-\frac{\sum_{j=1}^n(Ay_{t-1}+Ay_t)(j)}{n}\right)\right)^2 \label{online mirror descent 1 1}\\
&= \mu^2 \alpha_t^2 \sum_{i=1}^n (Ae_t(i)-\frac{\sum_{j=1}^n Ae_t(j)}{n} )^2. \label{online mirror descent 1 2}
    \end{align}
\end{subequations}  
Equality (\ref{online mirror descent 1 1}) is due to the update strategy of the row player in Equation (\ref{Update of online mirror descent}). Equality (\ref{online mirror descent 1 2}) is due to the existence of the fully-mixed equilibrium.
According to the Popoviciu's inequality we then have:
\begin{equation*}
 \frac{1}{n}\sum_{i=1}^n (Ae_t(i)-\frac{\sum_{j=1}^n Ae_t(j)}{n} )^2 \leq \frac{1}{4} (\argmax_{i \in \{1,\dots,m\}} Ae_t(i) -\argmin_{i \in \{1,\dots,m\}} Ae_t(i))^2 \leq \frac{1}{4} (1-0)^2 = 1/4.  
\end{equation*}
From the Algorithm \ref{LRCA algorithm 1} (LRCA) we have:
\[ \alpha_t = \min \left( \frac{f(x_{t-1})-v}{f(x_{t-1})}, \frac{f(x_{t-1})-v}{n/4} ,\frac{f(x_{t-1})-v}{2} \right) \leq \frac{x_{t-1}^\top Ae_t-v}{\mu\frac{n}{4}}.\] 
Follow the assumption that $\mu \leq 1$, we then have:  
\begin{equation*}
    \begin{aligned}
     &\|x_{t-1}-x^*\|^2 -\|x_{t+1}-x^*\|^2 \geq 2\mu (x_{t-1}-x^*)^\top (Ay_{t-1}+Ay_t) - \|x_{t-1}-x_{t+1}\|^2 \\
     &\geq 2\mu \alpha_t (x_{t-1}^\top Ae_t-v)- \frac{n}{4}\mu^2 \alpha_t^2\\
     &\geq \alpha_t (x_{t-1}^\top Ae_t-v) \geq 0.
    \end{aligned}
\end{equation*}
Now, using the same argument as in the proof of Theorem \ref{LRCA algorithm main proof} in the case of multiplicative weight update, we will have $\lim_{t \to \infty} x_t ={\bar{x}}^*$ where ${\bar{x}}^*$ is a minimax equilibrium strategy (the only difference is replacing $RE(x^*\|x)$ with $\|x^*-x\|^2$).

\subsection{Proof of Theorem \ref{proof for LMWU}}
From the step size assumption of LMWU algorithm, we have:
\[\exists t \in \mathbb{N}\; \text{such that}\; \mu_t \leq \frac{1}{3} \; \text{and}\; \lim_{i=t}^\infty \mu_i= \infty.\]
Using the update rule of LMWU in Definition \ref{LMWU definition} we obtain
\begin{equation*}
\begin{aligned}
    &\frac{x_{m+1}(1)}{x_m(1)}:\ldots:\frac{x_{m+1}(n)}{x_m(n)} =(1-\mu_m {e_1}^\top Ay_m):\ldots:(1-\mu_m {e_n}^\top Ay_m) \; \forall m.
\end{aligned}
\end{equation*}
Take $m$ equal $t$ and $t-1$ and time the equations side by side we obtain
\begin{equation*}
    \begin{aligned}
     &\frac{x_{t+1}(1)}{x_{t-1}(1)}:\frac{x_{t+1}(2)}{x_{t-1}(2)}:\ldots:\frac{x_{t+1}(n)}{x_{t-1}(n)}=  (1-\mu_{t} {e_1}^\top Ay_t)(1-\mu_{t-1} {e_1}^\top Ay_{t-1}):\\
     &(1-\mu_t {e_2}^\top Ay_t)(1-\mu_{t-1} {e_2}^\top Ay_{t-1}):\ldots: (1-\mu_t {e_n}^\top Ay_t)(1-\mu_{t-1} {e_n}^\top Ay_{t-1})  \\
     &\implies x_{t+1}(i) = \frac{x_{t-1}(i) (1-\mu_t {e_i}^\top Ay_t)(1-\mu_{t-1} {e_i}^\top Ay_{t-1}) }{\sum_{j=1}^n x_{t-1}(j) (1-\mu_t {e_j}^\top Ay_t)(1-\mu_{t-1} {e_j}^\top Ay_{t-1})} \quad \forall i \in {1,2,\dots n}.
    \end{aligned}
\end{equation*}
Note that for $t$ is event, $y_{t-1}=y^*$ in LRCA-\ref{LRCA algorithm 1} algorithm. For any $i$ such that : ${e_i}^\top Ay^*=v$ we have:

\begin{equation*}
\begin{aligned}
    \frac{x_{t+1}(i)}{x_{t-1}(i)} &= \frac{(1-\mu_{t-1} {e_i}^\top Ay^*)(1-\mu_t {e_i}^\top Ay_{t}) }{\sum_{j=1}^n x_{t-1}(j) (1-\mu_{t-1} {e_j}^\top Ay^*)(1-\mu_t {e_j}^\top Ay_{t})}\\
    &= \frac{(1-\mu_{t-1} v)(1-\mu_t {e_i}^\top Ay_{t}) }{\sum_{j=1}^n x_{t-1}(j) (1-\mu_{t-1} {e_j}^\top Ay^*)(1-\mu_t {e_j}^\top Ay_{t})}\\
    &= \frac{(1-\mu_t {e_i}^\top Ay_{t}) }{\sum_{j=1}^n x_{t}(j) \frac{1-\mu_{t-1} {e_j}^\top Ay^*}{1-\mu_{t-1} v}(1-\mu_t e_j^\top Ay_t)} \geq \frac{(1-\mu_t {e_i}^\top Ay_{t}) }{\sum_{j=1}^n x_{t-1}(j)(1-\mu_t {e_j}^\top Ay_{t})}.\\
\end{aligned}
\end{equation*}
The last inequality is due to $e_j^\top Ay^* \geq v \;\; \forall j \in \{1,\dots,n\}$.

We also have for any j such that : ${e_j}^\top Ay^* > v$ then $x^*(j)=0$ for any minimax equilibrium strategy $x^*$. Therefore, we have:
\begin{equation*}
    \begin{aligned}
        RE(x^*\|x_{t-1})-RE(x^*\|x_{t+1}) = \sum_{i=1}^n x^*(i) log\left(\frac{x_{t+1}(i)}{x_{t-1}(i)}\right)\\
        \geq \sum_{i=1}^n x^*(i) \log\left(\frac{(1-\mu_t {e_i}^\top Ay_{t}) }{\sum_{j=1}^n x_{t-1}(j)(1-\mu_t {e_j}^\top Ay_{t})}\right) \\
        = \sum_{i=1}^n x^*(i) \log\left(\frac{(1-\mu_t {e_i}^\top Ay_{t}) }{1-\mu_t x_{t-1}^\top Ay_t}\right).
        \end{aligned}
\end{equation*}
Applying inequality $log(x) \geq (x-1)-(x-1)^2\; \forall x \geq 0.5$ to the above equation, we obtain 
\begin{equation*}
    \begin{aligned}
       RE(x^*\|x_{t-1})-RE(x^*\|x_{t+1}) \geq \sum_{i=1}^n x^*(i) \left(\frac{(1-\mu_t {e_i}^\top Ay_{t}) }{1-\mu_t x_{t-1}^\top Ay_t}-1 - \left(\frac{(1-\mu_t {e_i}^\top Ay_{t}) }{1-\mu_t x_{t-1}^\top Ay_t}-1\right)^2 \right) \\
       = \frac{\mu_t (x_{t-1}^\top Ay_t-{x^*}^\top Ay_t)}{1-\mu_t x_{t-1}^\top Ay_t} - \sum_{i=1}^n x^*(i) \frac{\mu_t^2 (x_{t-1}^\top Ay_t-{e_i}^\top Ay_t)^2}{(1-\mu_t x_{t-1}^\top Ay_t)^2}.
    \end{aligned}
\end{equation*}
Now, follow the Algorithm \ref{LRCA algorithm 1} (LRCA), we have: $y_t= (1-\alpha_t)y^*+ \alpha_t e_t$. For j such that $e_j^\top Ay^* >v$, we have $x^*(j)=0$. We can simplify the above equation accordingly and use the Cauchy theorem to obtain
\begin{equation} \label{LMWU proof equation 1}
    \begin{aligned}
      &RE(x^*\|x_{t-1})-RE(x^*\|x_{t+1}) \geq \\
      &\frac{\mu_t (1-\alpha_t) (x_{t-1}^\top Ay^*-v)}{1-\mu_t x_{t-1}^\top Ay_t}-\sum_{i=1}^n x^*(i)\frac{2\mu_t^2(1-\alpha_t)^2 (x_{t-1}^\top Ay^*-v)^2}{(1-\mu_t x_{t-1}^\top Ay_t)^2}\\
      &+\frac{\mu_t \alpha_t (x_{t-1}^\top Ae_t-{x^*}^\top Ae_t)}{1-\mu_t x_{t-1}^\top Ay_t} -  \sum_{i=1}^n x^*(i)\frac{2\mu_t^2\alpha_t^2 (x_{t-1}^\top Ae_t-e_i^\top Ae_t)^2}{(1-\mu_t x_{t-1}^\top Ay_t)^2}.
    \end{aligned}
\end{equation}
For $\mu_t \leq \frac{1}{3}$ we have:\[\frac{\mu_t (1-\alpha_t) (x_{t-1}^\top Ay^*-v)}{1-\mu_t x_{t-1}^\top Ay_t}-\sum_{i=1}^n x^*(i)\frac{2\mu_t^2(1-\alpha_t)^2 (x_{t-1}^\top Ay^*-v)^2}{(1-\mu_t x_{t-1}^\top Ay_t)^2} \geq 0.\]
We also have:\[ \frac{(x_{t-1}^\top Ae_t-e_i^\top Ae_t)^2}{(1-\mu_t x_{t-1}^\top Ay_t)^2} \leq \frac{1}{(1-\mu_t)(1-\mu_t x_{t-1}^\top Ay_t)}.\]
Follow the Inequality (\ref{LMWU proof equation 1}), we obtain 
\begin{equation*}
    \begin{aligned}
    RE(x^*\|x_{t-1})-RE(x^*\|x_{t+1}) \geq \frac{\mu_t \alpha_t (x_{t-1}^\top Ae_t-{x^*}^\top Ae_t)}{1-\mu_t x_{t-1}^\top Ay_t} -\frac{2\mu_t^2\alpha_t^2}{(1-\mu_t)(1-\mu_t x_{t-1}^\top Ay_t)}.
    \end{aligned}
\end{equation*}
By definition of $\alpha_t$ in LRCA-\ref{LRCA algorithm 1} algorithm  \[\alpha_t \leq \frac{x_{t-1}^\top Ae_t-v}{2} \leq \frac{x_{t-1}^\top Ae_t-{x^*}^\top Ae_t}{2},\] along with $\mu_t \leq \frac{1}{3}$ we have:
\[\frac{1}{2}\frac{\mu_t \alpha_t (x_{t-1}^\top Ae_t-{x^*}^\top Ae_t)}{1-\mu_t x_{t-1}^\top Ay_t} \geq \frac{2\mu_t^2\alpha_t^2}{(1-\mu_t)(1-\mu_t x_{t-1}^\top Ay_t)}.\]
Thus, we have:
\begin{equation*}
    \begin{aligned}
     RE(x^*\|x_{t-1})-RE(x^*\|x_{t+1}) \geq \frac{1}{2}\frac{\mu_t \alpha_t (x_{t-1}^\top Ae_t-{x^*}^\top Ae_t)}{1-\mu_t x_{t-1}^\top Ay_t}\\
     \geq \frac{1}{2}\frac{\mu_t \alpha_t (x_{t-1}^\top Ae_t-v)}{1-\mu_t x_{t-1}^\top Ay_t} \geq \frac{\mu_t \alpha_t (x_{t-1}^\top Ae_t-v)}{2} \geq 0 .  
    \end{aligned}
\end{equation*}
Now, using the same argument in the Theorem \ref{LRCA algorithm main proof}, we will have $\lim_{t \to \infty} x_t =x^*$.
\subsection{Proof of Theorem \ref{OMWU proof}} \label{proof of MWU detail}

The LRCA algorithm with 2 stability factors is described as follow:
\begin{algorithm}[t]
\SetAlgoLined
\textbf{Input:} Current iteration $t$, past feedback $x_{t-1}^\top A$ of the row player\\
\textbf{Output:} Strategy $y_t$ for the column player\\
\If{$t=3k-1\; or\; 3k, \; k \in \mathbb{N}$}{$y_t=y^*$}
\If{$t=3k+1, \; k \in \mathbb{N}$}{ $e_t : = \argmax_{e \in \{e_1,e_2,\dots e_n\}}x_{t-1}^\top Ae$\\
$f(x_{t-1})= \max_{y \in \Delta_m}x_{t-1}^\top Ay$\\
$\alpha_t= \frac{f(x_{t-1})-v}{f(x_{t-1})}$\\
$y_t=(1-\alpha_t)y^*+\alpha_t e_t$ \\} 
\caption{LRCA-2 algorithm with two stabilizing strategy} \label{LRCA algorithm with 2 stablizing}
\end{algorithm}

Follow the update rule of OMWU, we have:
\begin{equation*}
    \begin{aligned}
        x_{3k+3}(i)=x_{3k}(i)\frac{e^{-2\mu e_i^\top Ay_{3k+2}-\mu e_i^\top Ay_{3k+1}-\mu e_i^\top Ay_{3k}+\mu e_i^\top Ay_{3k-1}}}{ \sum_{j=1}^n x_{3k}(j) e^{-2\mu e_j^\top Ay_{3k+2}-\mu e_j^\top Ay_{3k+1}-\mu e_j^\top Ay_{3k}+\mu e_j^\top Ay_{3k-1}} }.
    \end{aligned}
\end{equation*}
We then derive:
\begin{equation*}
    \begin{aligned}
        &RE(x^*\|x_{3k})-RE(x^*\|x_{3k+3}) = \sum_{i=1}^n x^*(i) \log \left(\frac{x_{3k+3}(i)}{x_{3k}(i)}\right)\\
       & = \sum_{i=1}^n x^*(i) \log\left(\frac{e^{-2\mu e_i^\top Ay^*-\mu e_i^\top Ay_{3k+1}}}{\sum_{j=1}^n x_{3k}(j)e^{-2\mu e_j^\top Ay^*-\mu e_j^\top Ay_{3k+1}}}\right)\\
       & =-2 \mu v -\mu {x^*}^\top Ay_{3k+1}-\log\left(\sum_{j=1}^n x_{3k}(j)e^{-2\mu e_j^\top Ay^*-\mu e_j^\top Ay_{3k+1}}\right)\\
        &\geq -2\mu v - \mu v - (-2\mu v) - \log\left(\sum_{j=1}^n x_{3k}(j) e^{-\mu e_j^\top Ay_{3k+1}}\right)\\
       & = -\mu v - \log\left(\sum_{j=1}^n x_{3k}(j) e^{-\mu e_j^\top Ay_{3k+1}}\right),
    \end{aligned}
\end{equation*}
where the inequality comes from the property of the minimax equilibrium:
\[{x^*}^\top Ay \leq v \; \forall y \in \Delta_m \;\; ; \;\; x^\top Ay^* \geq v \; \forall x \in \Delta_n.\]
It then comes to the exact formulation of Inequality (\ref{MWU1a}) in proof of Lemma \ref{MWU K-L distance Lemma}. By choosing the same step size as Theorem \ref{LRCA algorithm main proof}, we will have the last round convergence. 
\subsection{Proof of Theorem \ref{no-instant-regret theorem}} \label{no-instant-regret proof}
We continue the proof in the case of OMD/FTRL. Following the Theorem \ref{OMD/FTRL} we have
\begin{equation*}
    \begin{aligned}
     \sum_{k=1}^{T/2} (f(x_{2k-1})-v)^2 \leq \frac{n}{4} 2 \\
     \implies \sum_{k=1}^{T/2} (f(x_{2k-1})-v) \leq \frac{1}{2} T^{1/2} n^{1/2},
    \end{aligned}
\end{equation*}
since $||x_t-x^*||^2 \leq 2$. Using the same argument as the case of MWU, we then have:
\[IR_T \leq T^{1/2} n^{1/2}.\]
When the row player using LMWU, using the result in Theorem \ref{proof for LMWU}, the proof is exactly the same as in the case of MWU.

\end{document}